%% file: root.tex
\documentclass{article}

\usepackage[margin=1in]{geometry}
\usepackage{graphicx}
\usepackage{amsmath, amssymb, amsthm}
\usepackage{xcolor}
\usepackage{authblk}

\newcommand{\SO}[1]{\ensuremath{\textrm{SO}(#1)} }
\newcommand{\SE}[1]{\ensuremath{\textrm{SE}(#1)} }
\newcommand{\se}[1]{\ensuremath{\mathfrak{se}(#1)} }

\newcommand{\dotpr}[2]{\ensuremath{\langle #1, #2 \rangle }}

\def\tr{\ensuremath{\textrm{trace}} }
\def\R{\ensuremath{\mathbb{R}} }
\def\G{\ensuremath{G} }
\def\g{\ensuremath{\mathfrak{g}} }
\def\Pa{\ensuremath{\pi_\g} }

\newtheorem{thm}{Theorem}
\newtheorem{assum}{Assumption}

\graphicspath{{./Figures/}}

\begin{document}

\title{Globally exponentially convergent observer for \\ systems evolving on matrix Lie groups}
\date{}
\author{Soham Shanbhag\thanks{sshanbhag@kaist.ac.kr} }
\author{Dong Eui Chang\thanks{Corresponding author, dechang@kaist.ac.kr}}
\affil{School of Electrical Engineering, Korea Advanced Institute of Science and Technology, Daejeon, Republic of Korea}

\maketitle
\begin{abstract}
    We propose a globally exponentially convergent observer for the dynamical system evolving on matrix Lie groups with bounded velocity with unknown bound.
    We design the observer in the ambient Euclidean space and show exponential convergence of the observer to the state of the system.
    We show the convergence with an example of a rigid body rotation and translation system on the special Euclidean group.
    We compare the proposed observer with an observer present in the literature.

    \textbf{Keywords:} Continuous time observer, Lie groups
\end{abstract}


\section{Introduction}\label{sec_intro}
In design of feedback based control systems, the estimates of the state are used to compute the control input given to the system.
Often, the measurements of the system are not the true measurements of the system, or are biased and contain noise.
For this purpose, estimators are designed for the system in question to provide estimates of the state and the bias based on measurements and the underlying system model.
The design of such estimators is hence an important area of research.

Previous research in the field of observer design for systems on Lie groups is restricted to designing observers on the Lie groups.
The authors in \cite{KhosrTML2015} design a gradient based uniformly locally exponential observer on Lie groups.
Similar local exponential results are achieved by \cite{LagemTM2010} on Lie groups.
The results in case of matrix Lie groups are much more numerous.
The authors in \cite{HuaZTMH2011} design a locally asymptotically stable observer for the system evolving on the Special Euclidean group.
The authors in \cite{WangT2020b} present a tutorial of various observers designed on the special Euclidean group.
However, in case of observer systems designed on Lie groups, these estimators are not able to provide estimates for all states due to the underlying group structure \cite{BhatB1998}.
To overcome this limitation, the authors in \cite{WangT2019} propose a hybrid structure of the observer on SE(3).
While this leads to global convergence of the observer, this leads to a discontinuous observer.
Another method for achieving global convergence results is by designing the observer in the ambient Euclidean space.
The author in \cite{Chang2021} designs a global exponential continuous observer for systems evolving on matrix Lie groups using this technique, but he assumes knowledge of bounds of the velocity of the system.
While the system in consideration may have a bounded velocity, the knowledge of the bounds is usually not known apriori.
Hence, without the knowledge of the bounds, an observer cannot be designed for the dynamical system using the results in \cite{Chang2021}.
Using a Kalman/Riccati observer for the same is a possible method with the above technique, however, it is computationally expensive(atleast $n^2$ dimensional, where $n$ is the dimension of the Euclidean space).

We propose a globally exponentially convergent observer for the dynamical system evolving on a Lie group.
We design this observer on the ambient Euclidean space to achieve global convergence properties, while being continuous.
This extension allows us to treat the considered system as a linear time varying system, and design a constant gain observer for the linear time varying system.
Moreover, due to no assumptions of the knowledge of bounds while achieving globality of results while being continuous, our observer improves upon existing works in the literature.

The observer proposed in \cite{ShanbC2022a} is a specific case of the observer proposed in this paper for the rigid body rotation and translation system.
The observer in \cite{ShanbC2022a} improves on the proposed observer in this paper by relaxing some of the bounds in this paper (linear velocity and position of the system).
However, these assumptions can only be relaxed due to the structure of the rigid body system, and hence the observer in \cite{ShanbC2022a} does not apply to other systems.


\section{Preliminaries}\label{sec_prelim}
We denote the estimate of the state $A$ by $\bar{A}$.
The Euclidean inner product of two matrices in $\R^{m \times n}$ is denoted by $\dotpr{\cdot}{\cdot}: \R^{m \times n} \times \R^{m \times n} \to \R$ such that $\dotpr{A}{B} = \tr(A^T B)$.
The Euclidean norm of a matrix $A \in \R^{m \times n}$ is defined as $\| A \| = \sqrt{\langle A, A \rangle}$.
The orthogonal projection of $A \in \R^{n \times n}$ to a vector space $V$, which is a subspace of $\R^{n \times n}$, is denoted by $
\pi_V (A)$.
We denote the $i^{th}$ singular value of matrix $A$ by $\sigma_i(A)$, where the subscript \verb|min| and \verb|max| denote the minimum and maximum singular values respectively.

Let $\G \subset \R^{n \times n}$ be a matrix Lie group, and $\g$ be its Lie algebra.
Consider the dynamical system in continuous time
\begin{align} \label{sys_main}
    \dot{g} &= g \xi,
\end{align}
where $g \in \G \subset \R^{n \times n}$ denotes the state of the system, and $\xi \in \g$ denotes the velocity of the body.
We assume that a matrix valued signal $A_m(t)$ is available such that $A_m = Fg$, where $F$ is a constant invertible matrix in $\R^{n \times n}$.
The measurement of velocity is assumed to be biased as $\xi_m = \xi + b$, where $b \in \g$ is a constant or a slowly varying bias.
This can be well approximated by the equation $\dot{b} = 0$.
This is a common form of measurement, as can be seen in \cite{KhosrTML2015,LagemTM2010,WangT2019,Chang2021} as well as the references contained therein.
All measurements are assumed to be available in continuous time.
In practice, the measurements of the states and the velocities will be corrupted by noise.

Moreover, consider the following assumptions on the velocity and state of the system:
\begin{assum}\label{assum:vel_bound}
    The velocity of the system, $\xi$, is bounded, i.e., there exists an unknown $M > 0$ such that $\sup_{t \ge 0} \| \xi(t) \| < M$.
\end{assum}
\begin{assum}\label{assum:traj}
    There exist unknown $L_g > 0$ and $R_g > 0$ such that $L_g \leq \sigma_{\textrm{min}}(g(t)) \leq \sigma_{\textrm{max}}(g(t)) \leq R_g$ for all $t \geq 0$.
\end{assum}
Since $\sigma_{\textrm{max}}(g(t)) \leq R_g$, we have that $\| g(t) \| = \sum_{i = 1}^n \sigma_i(g(t)) \leq n R_g$ is bounded for all $t \geq 0$.
Also, $\sigma_{\textrm{min}}(g(t)) \geq L_g$ ensures that $\| g^{-1}(t) \| = \sum_{i = 1}^n 1/\sigma_i(g(t)) \leq n/L_g$ is bounded for all $t \geq 0$.


\section{Proposed Observer}\label{sec_observer} 
We define the following observer system:
\begin{align}\label{sys_est}
    \begin{aligned}
        \dot{\bar{A}} &= A_m \xi_m + k_1 (A_m - \bar{A}) - A_m \bar{b},\\
        \dot{\bar{b}} &= - k_2 \Pa(A_m^T (A_m - \bar{A})),
    \end{aligned}
\end{align}
where the observer state $(\bar{A}(t), \bar{b}(t)) \in \R^{n \times n} \times \g$, and $k_1$ and $k_2$ are constants to be chosen by the designer of the estimator.

\begin{thm}\label{theo_main}
    Define the observer error terms
    \begin{align*}
        E_A = A - \bar{A}, \quad E_b = b - \bar{b},
    \end{align*}
    where $A = Fg$. Suppose that Assumptions \ref{assum:vel_bound} and \ref{assum:traj} hold.
    Then, with $k_1, k_2 > 0$ and the update laws as given in equations \eqref{sys_main} and \eqref{sys_est}, the error terms $\| E_A(t) \|$ and $\| E_b(t) \|$ converge exponentially to $0$ as $t \to \infty$ for all $(\bar{A}(0), \bar{b}(0)) \in \R^{n \times n} \times \g$.
\end{thm}

\begin{proof}
    Differentiating the error terms $E_A$ and $E_b$ with respect to time along the trajectory of equations \eqref{sys_main} and \eqref{sys_est}, we get the error system
    \begin{subequations}\label{sys_err}
    \begin{align}
        \dot{E}_A &= - k_1 E_A - A E_b\label{eq_err_EA},\\
        \dot{E}_b &= k_2 \Pa(A^T E_A) \label{eq_err_Eb}.
    \end{align}
    \end{subequations}
    The measurement models explained in section \ref{sec_prelim} have been substituted here for the measurement terms in equation \eqref{sys_est}.

    Define the Lyapunov function
    \begin{align}\label{eq_lyapunov}
        V = \frac{k_2}{2} \dotpr{E_A}{E_A} + \frac12 \dotpr{E_b}{E_b}.
    \end{align}
    Differentiating $V$ with respect to time along the trajectory of system \eqref{sys_err}, we get
    \begin{align*}
        \dot{V} &= k_2 \dotpr{E_A}{\dot{E}_A} + \dotpr{E_b}{\dot{E}_b}
        = k_2 \dotpr{E_A}{- k_1 E_A - A E_b} + \dotpr{E_b}{k_2 \Pa(A^T E_A)}\\
        &= - k_1 k_2 \dotpr{E_A}{E_A} - k_2 \dotpr{E_A}{A E_b} + k_2 \dotpr{E_b}{A^T E_A}
        = -k_1 k_2 \dotpr{E_A}{E_A},
    \end{align*}
    where we use the property $\dotpr{E_b}{\Pa(A^T E_A)} = \dotpr{E_b}{A^T E_A}$ due to orthogonality of $E_b \in \g$ and $A^T E_A - \Pa(A^T E_A) \in \g^\perp$, the orthogonal complement of $\g$.
    Since $\dotpr{E_A}{E_A}$ is non-negative, $\dot{V}$ is negative semi-definite.
    Hence, $\dotpr{E_A}{E_A}$ and $\dotpr{E_b}{E_b}$ are bounded.
    Moreover, since $V$ is bounded below by $0$ and non-increasing, $V(t)$ has a finite limit as $t \to \infty$.

    To prove asymptotic stability, we use Barbalat's lemma. Differentiating $\dot{V}$ with respect to time along the trajectory of system \eqref{sys_err}, we get
    \begin{align*}
        \ddot{V} &= -2 k_1 k_2 \dotpr{E_A}{\dot{E}_A} = -2 k_1 k_2 \dotpr{E_A}{- k_1 E_A - A E_b}
        = 2 k_1^2 k_2 \dotpr{E_A}{E_A} + 2 k_1 k_2 \dotpr{E_A}{A E_b}.
    \end{align*}
    Note that $\dotpr{E_A}{E_A}$ and $\dotpr{E_b}{E_b}$ are bounded from the initial value and decreasing nature of $V(t)$.
    Since $g$ is bounded due to Assumption \ref{assum:traj} along the trajectory of system \eqref{sys_main}, $A$ is bounded.
    Hence, the inner product $|\dotpr{E_A}{A E_b}| \leq \| E_A \| \| A \| \| E_b \|$ is also bounded.
    This implies that $\ddot{V}$ is bounded.
    Consequently $\dot{V}$ is uniformly continuous.
    By Barbalat's lemma, since $V$ has a finite limit as limit as $t \to \infty$, and $\dot{V}$ is uniformly continuous, $\dot{V}(t) \to 0$ as $t \to \infty$.
    Hence, $\lim_{t \to \infty}\| E_A(t) \| = 0$.
    Consequently, $\lim_{t \to \infty} E_A(t)  = 0$.

    To prove asymptotic convergence of $E_b$ to $0$, consider equation \eqref{eq_err_EA}.
    Differentiating equation \eqref{eq_err_EA} with respect to time along the trajectory of system \eqref{sys_err},
    \begin{align*}
        \ddot{E}_A &= - k_1 \dot{E}_b - \dot{A} E_b - A \dot{E}_b
        = k_1^2 E_A + k_1 A E_b - A \xi E_b - k_2 A \Pa(A^T E_A),
    \end{align*}
    which is bounded since $E_A$ and $E_b$ are bounded due to boundedness of $V$, $A$ is bounded due to boundedness of $g$ by Assumption \ref{assum:traj} along the trajectory of system \eqref{sys_main}, and $\xi$ is bounded due to Assumption \ref{assum:vel_bound}.
    Hence, $\dot{E}_A$ is uniformly continuous.
    From Barbalat's lemma, since $\lim_{t \to \infty} E_A(t) = 0$, we have that $\lim_{t \to \infty} \dot{E}_A(t) = 0$.
    Since $\G$ is a matrix Lie group, $g^{-1} \in \G$ exists.
    Hence, $A^{-1}(t) = g^{-1}(t) F^{-1}$ exists, is non-zero and finite due to Assumption \ref{assum:traj}.
    From equation \eqref{eq_err_EA}, $\lim_{t \to \infty} E_b(t) = \lim_{t \to \infty} - A^{-1}(t) (\dot{E}_A(t) + k_1 E_A(t)) = 0$.
    Hence, the error system is asymptotically stable.

    Since $V$ defined in equation \eqref{eq_lyapunov} is independent of $t$, the system is uniformly asymptotically stable.
    Using the fact that the error system \eqref{sys_err} is linear, from Theorem 4.11 in \cite{Khali2002}, we conclude that it is globally exponentially stable.
\end{proof}


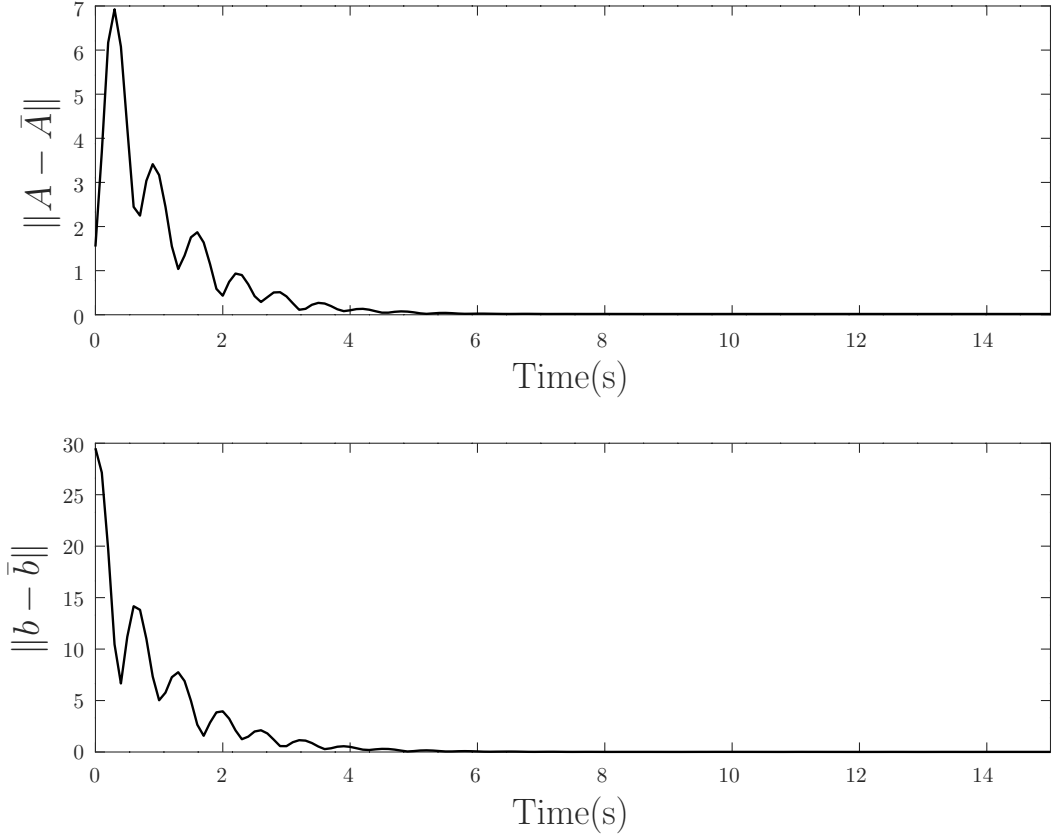
\begin{figure}[!b]
    \centering
    \resizebox{\linewidth}{!}{\input{Figures/continuous_LG_err_noise.tex}}
    \caption{Simulation of the proposed observer in presence of noise\label{fig_sim_SE3_noise}}
\end{figure}

\section{Numerical Simulation}\label{sec_sim}
To demonstrate the convergence of the proposed observer, we consider a rigid body translational and rotational dynamics system evolving on the special Euclidean group.
Define $\SO{3}$ as the special orthogonal group in three dimensions such that $R \in \SO{3}$ implies $R^T R = I$ and $\det(R) = 1$.
Let $\SE{3}$ denote the special Euclidean group defined as
\begin{align*}
    \SE{3} =  \left\{ \begin{bmatrix} R & p \\ 0 & 1 \end{bmatrix} \in \R^{4 \times 4} \mid R \in \SO{3}, p \in \R^3 \right\}.
\end{align*}
Define the hat map $\hat{\cdot}: \R^3 \to \R^{3\times 3}$ such that $\forall ~ x, y \in \R^3, x \times y = \hat{x} y$.
Then, the Lie algebra of $\SE{3}$ is denoted by
\begin{align*}
    \se{3} = \left\{ \begin{bmatrix} \hat{\Omega} & v \\ 0 & 0 \end{bmatrix} \in \R^{4 \times 4} \mid \Omega \in \R^{3}, v \in \R^3 \right\}.
\end{align*}
Since this is a vector space, the projection onto $\se{3}$ of a matrix in $\R^{4 \times 4}$ is defined as
\begin{align*}
    \pi_{\se{3}} \left(
    \begin{bmatrix}
        A & b \\ c^T & d
    \end{bmatrix}
    \right) =
    \begin{bmatrix}
        \frac12 (A - A^T) & b \\ 0 & 0
    \end{bmatrix},
\end{align*}
where $A \in \R^{3 \times 3}, b\in \R^3, c \in \R^3$ and $d \in \R$.

The system equation can then be defined as
\begin{align}\label{sys_SE3}
    \dot{X} = X \xi,
\end{align}
where $X \in \SE{3}$ is the state of the body, and $\xi \in \se{3}$ denotes the velocity of the body.
For measurement, landmarks at $x \in \R^3$ in homogeneous coordinates are denoted as $(x, 1)$, while such vectors at infinity (for example, the gravity vector) along $x \in \R^3$ are denoted as $(x, 0)$.
Suppose that for measurement, we measure the inertial vectors $s_1 = (e_1, 1), s_2 = (e_2, 1), s_3 = (e_3, 1)$ and $s_4 = (-e_3, 0)$. Define $F = \begin{bmatrix}s_1 & s_2 & s_3 & s_4\end{bmatrix}$.

Consider the true velocities of the system to be $\Omega = (-\sin(10t), \cos(10t), 0)$ and $v = (\cos(0.5t), \sin(0.5t), 0 )$.
Let the bias in the measurements of the velocities be $b_\Omega = (-10, 15, 8)$, and $b_v = (2, 8, 5)$,
such that the total bias is given by
\begin{align*}
    b = \begin{bmatrix} \hat{b}_\Omega & b_v \\ 0 & 0 \end{bmatrix}.
\end{align*}
We also assume the presence of Gaussian noise in measurements with mean $0$ and standard deviation $0.1$, such that the measurements are of the form $R_m = R \exp(\hat{w})$, $p_m = p + v$, and $\xi_m = \xi + b + z$, where $w, v \in \R^3$ and $z \in \se{3}$ are random numbers with mean $0$ and standard deviation $0.1$ to simulate a sensor.
This system leads to all assumptions being valid.

We design an observer based on Theorem \ref{theo_main} for the system evolving on \eqref{sys_SE3}.
The system trajectory is initialized with the state $R = I, p = (0, 0, 1)$.
The observer is initialized with the state $\bar{R} = \exp(-\pi \hat{e}_3/10)$, $\bar{p} = 0$ and $\bar{b} = 0$.
The simulation is run for 15 seconds.
We choose $k_1 = 2$ and $k_2 = 10$ for the simulation.
The error between the observer state and the system state as a function of time is shown in Figure \ref{fig_sim_SE3_noise}.
As can be seen from the simulations, the proposed observer converges fairly quickly to the true state.
Also, the bias in velocity measurements is estimated correctly.
Simulation in presence of noise shows that the observer is robust to noise.

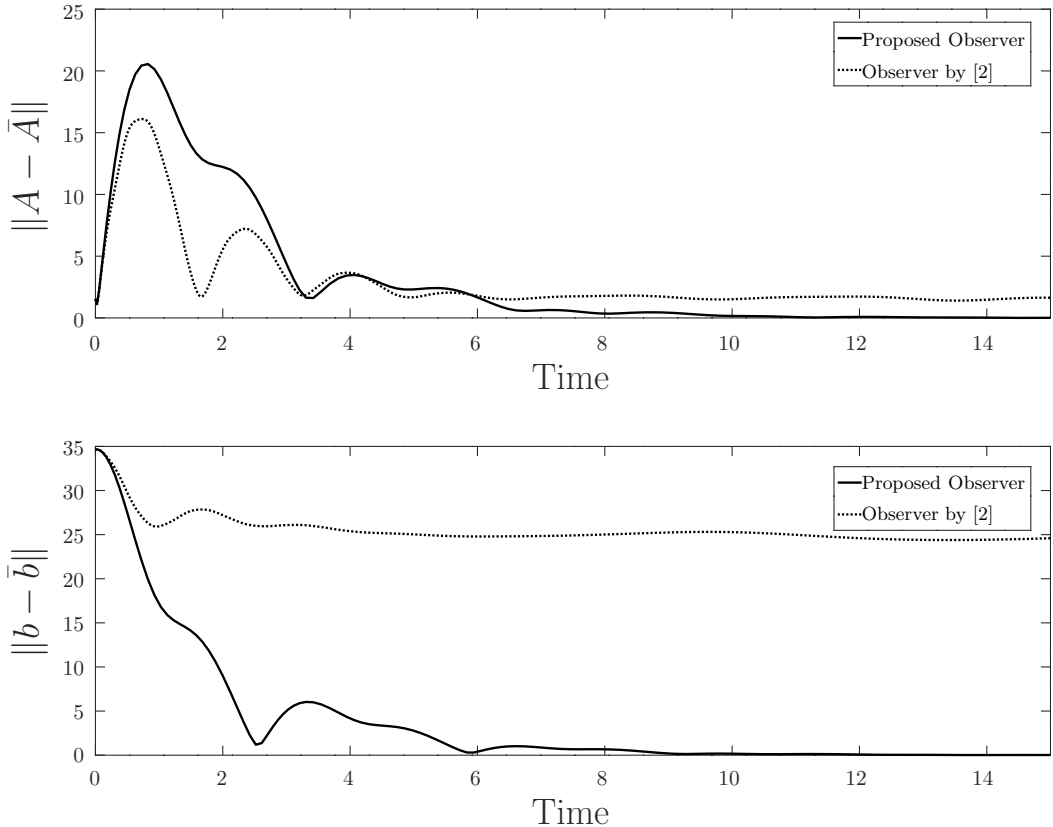
\begin{figure}[!b]
    \centering
    \resizebox{\linewidth}{!}{\input{Figures/continuous_LG_err_comp.tex}}
    \caption{Comparison of proposed observer with literature\label{fig_sim_SE3_err_comp}}
\end{figure}

To compare the proposed observer with previously available observer in literature, we consider the observer proposed by \cite{Chang2021}.
The observer in \cite{Chang2021} is a good candidate observer for comparison with the proposed observer in this paper since it also provides a continuous and globally convergent observer.
Other observers in literature are either not globally convergent, or discontinuous observers.
We choose $k_1 = 1$ and $k_2 = 1$ for the observers.
The velocity profile for the observers are chosen as $\Omega = 0$ and $v = (\cos(t), \sin(t), 0.5 \sin(2t))$, while the bias profiles are chosen as $b_\Omega = (10, 10, 10)$ and $b_v = (10, 20, 10)$.
Note that, these values are intended to violate the $k_1 > \sup_{t \geq 0}\| \xi(t) \| + \| b \|$ condition assumed in \cite{Chang2021}, since during application, the designer of the controller wouldn't know the magnitude of the bias.
This simulation is run for 25 seconds.
The simulation of this observer can be seen in Figure \ref{fig_sim_SE3_err_comp}.
As can be seen in the simulation, the observer proposed in \cite{Chang2021} fails to converge, while the observer proposed in this paper converges fairly quickly.
Moreover, the bias estimates of the observer in \cite{Chang2021} also fail to converge to the correct values, as compared to the proposed observer.


\section{Conclusion}\label{sec_conc}
A globally exponentially convergent observer is designed for the dynamical system on matrix Lie groups.
The observer improves upon existing observers by being continuous and globally convergent in its estimate, which is a feature most observers designed on Lie groups fail at.
We improve on observers showing global convergence in the literature by not assuming knowledge of velocity bounds.
A possible future research topic would be to design global exponential observers without any assumption of boundedness of the velocities.


\bibliographystyle{plain}
\bibliography{References}

\end{document}

%% file: Figures/continuous_LG_err_noise.tex
\setlength{\unitlength}{1pt}
\begin{picture}(0,0)
\includegraphics{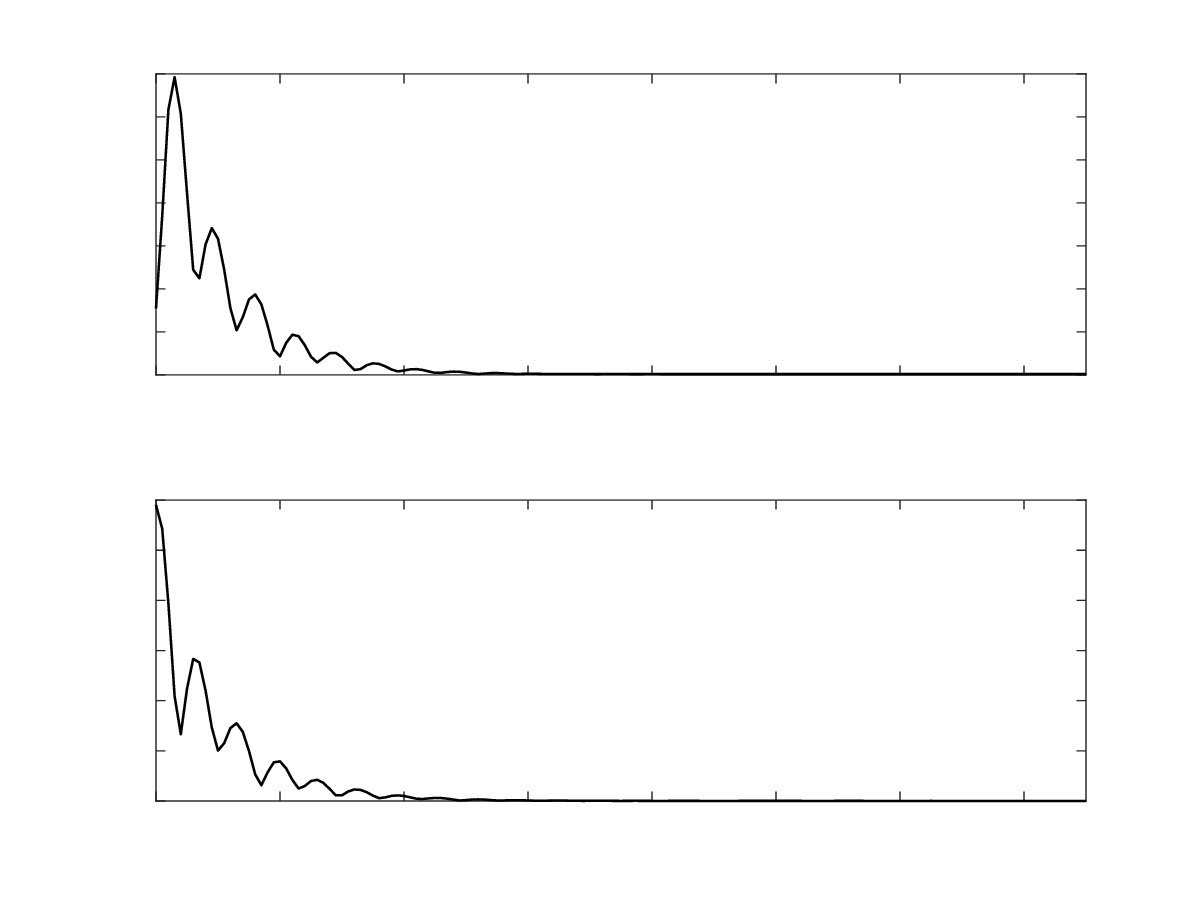}
\end{picture}%
\begin{picture}(576,432)(0,0)
\fontsize{10}{0}
\selectfont\put(74.8799,244.537){\makebox(0,0)[t]{\textcolor[rgb]{0.15,0.15,0.15}{{0}}}}
\fontsize{10}{0}
\selectfont\put(134.4,244.537){\makebox(0,0)[t]{\textcolor[rgb]{0.15,0.15,0.15}{{2}}}}
\fontsize{10}{0}
\selectfont\put(193.92,244.537){\makebox(0,0)[t]{\textcolor[rgb]{0.15,0.15,0.15}{{4}}}}
\fontsize{10}{0}
\selectfont\put(253.44,244.537){\makebox(0,0)[t]{\textcolor[rgb]{0.15,0.15,0.15}{{6}}}}
\fontsize{10}{0}
\selectfont\put(312.96,244.537){\makebox(0,0)[t]{\textcolor[rgb]{0.15,0.15,0.15}{{8}}}}
\fontsize{10}{0}
\selectfont\put(372.48,244.537){\makebox(0,0)[t]{\textcolor[rgb]{0.15,0.15,0.15}{{10}}}}
\fontsize{10}{0}
\selectfont\put(432,244.537){\makebox(0,0)[t]{\textcolor[rgb]{0.15,0.15,0.15}{{12}}}}
\fontsize{10}{0}
\selectfont\put(491.52,244.537){\makebox(0,0)[t]{\textcolor[rgb]{0.15,0.15,0.15}{{14}}}}
\fontsize{10}{0}
\selectfont\put(69.8755,252.06){\makebox(0,0)[r]{\textcolor[rgb]{0.15,0.15,0.15}{{0}}}}
\fontsize{10}{0}
\selectfont\put(69.8755,272.694){\makebox(0,0)[r]{\textcolor[rgb]{0.15,0.15,0.15}{{1}}}}
\fontsize{10}{0}
\selectfont\put(69.8755,293.328){\makebox(0,0)[r]{\textcolor[rgb]{0.15,0.15,0.15}{{2}}}}
\fontsize{10}{0}
\selectfont\put(69.8755,313.962){\makebox(0,0)[r]{\textcolor[rgb]{0.15,0.15,0.15}{{3}}}}
\fontsize{10}{0}
\selectfont\put(69.8755,334.597){\makebox(0,0)[r]{\textcolor[rgb]{0.15,0.15,0.15}{{4}}}}
\fontsize{10}{0}
\selectfont\put(69.8755,355.231){\makebox(0,0)[r]{\textcolor[rgb]{0.15,0.15,0.15}{{5}}}}
\fontsize{10}{0}
\selectfont\put(69.8755,375.865){\makebox(0,0)[r]{\textcolor[rgb]{0.15,0.15,0.15}{{6}}}}
\fontsize{10}{0}
\selectfont\put(69.8755,396.5){\makebox(0,0)[r]{\textcolor[rgb]{0.15,0.15,0.15}{{7}}}}
\fontsize{16}{0}
\selectfont\put(298.08,231.537){\makebox(0,0)[t]{\textcolor[rgb]{0.15,0.15,0.15}{{Time(s)}}}}
\fontsize{16}{0}
\selectfont\put(58.8755,324.279){\rotatebox{90}{\makebox(0,0)[b]{\textcolor[rgb]{0.15,0.15,0.15}{{$\| A - \bar{A} \|$}}}}}
\fontsize{10}{0}
\selectfont\put(74.8799,39.9971){\makebox(0,0)[t]{\textcolor[rgb]{0.15,0.15,0.15}{{0}}}}
\fontsize{10}{0}
\selectfont\put(134.4,39.9971){\makebox(0,0)[t]{\textcolor[rgb]{0.15,0.15,0.15}{{2}}}}
\fontsize{10}{0}
\selectfont\put(193.92,39.9971){\makebox(0,0)[t]{\textcolor[rgb]{0.15,0.15,0.15}{{4}}}}
\fontsize{10}{0}
\selectfont\put(253.44,39.9971){\makebox(0,0)[t]{\textcolor[rgb]{0.15,0.15,0.15}{{6}}}}
\fontsize{10}{0}
\selectfont\put(312.96,39.9971){\makebox(0,0)[t]{\textcolor[rgb]{0.15,0.15,0.15}{{8}}}}
\fontsize{10}{0}
\selectfont\put(372.48,39.9971){\makebox(0,0)[t]{\textcolor[rgb]{0.15,0.15,0.15}{{10}}}}
\fontsize{10}{0}
\selectfont\put(432,39.9971){\makebox(0,0)[t]{\textcolor[rgb]{0.15,0.15,0.15}{{12}}}}
\fontsize{10}{0}
\selectfont\put(491.52,39.9971){\makebox(0,0)[t]{\textcolor[rgb]{0.15,0.15,0.15}{{14}}}}
\fontsize{10}{0}
\selectfont\put(69.8755,47.52){\makebox(0,0)[r]{\textcolor[rgb]{0.15,0.15,0.15}{{0}}}}
\fontsize{10}{0}
\selectfont\put(69.8755,71.5933){\makebox(0,0)[r]{\textcolor[rgb]{0.15,0.15,0.15}{{5}}}}
\fontsize{10}{0}
\selectfont\put(69.8755,95.6665){\makebox(0,0)[r]{\textcolor[rgb]{0.15,0.15,0.15}{{10}}}}
\fontsize{10}{0}
\selectfont\put(69.8755,119.74){\makebox(0,0)[r]{\textcolor[rgb]{0.15,0.15,0.15}{{15}}}}
\fontsize{10}{0}
\selectfont\put(69.8755,143.813){\makebox(0,0)[r]{\textcolor[rgb]{0.15,0.15,0.15}{{20}}}}
\fontsize{10}{0}
\selectfont\put(69.8755,167.887){\makebox(0,0)[r]{\textcolor[rgb]{0.15,0.15,0.15}{{25}}}}
\fontsize{10}{0}
\selectfont\put(69.8755,191.96){\makebox(0,0)[r]{\textcolor[rgb]{0.15,0.15,0.15}{{30}}}}
\fontsize{16}{0}
\selectfont\put(298.08,26.9971){\makebox(0,0)[t]{\textcolor[rgb]{0.15,0.15,0.15}{{Time(s)}}}}
\fontsize{16}{0}
\selectfont\put(52.8755,119.74){\rotatebox{90}{\makebox(0,0)[b]{\textcolor[rgb]{0.15,0.15,0.15}{{$\| b - \bar{b} \|$}}}}}
\end{picture}

%% file: Figures/continuous_LG_err_comp.tex
\setlength{\unitlength}{1pt}
\begin{picture}(0,0)
\includegraphics{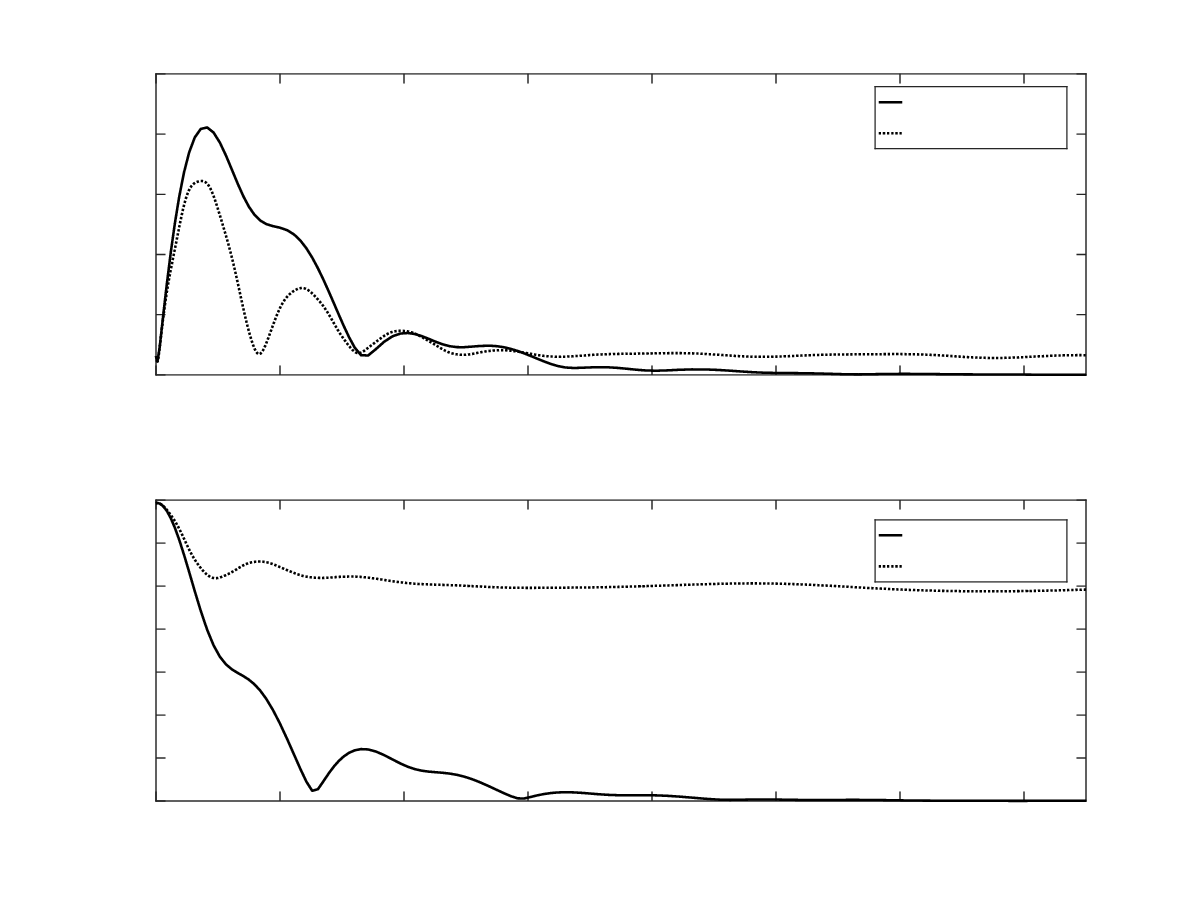}
\end{picture}%
\begin{picture}(576,432)(0,0)
\fontsize{10}{0}
\selectfont\put(74.8799,244.537){\makebox(0,0)[t]{\textcolor[rgb]{0.15,0.15,0.15}{{0}}}}
\fontsize{10}{0}
\selectfont\put(134.4,244.537){\makebox(0,0)[t]{\textcolor[rgb]{0.15,0.15,0.15}{{2}}}}
\fontsize{10}{0}
\selectfont\put(193.92,244.537){\makebox(0,0)[t]{\textcolor[rgb]{0.15,0.15,0.15}{{4}}}}
\fontsize{10}{0}
\selectfont\put(253.44,244.537){\makebox(0,0)[t]{\textcolor[rgb]{0.15,0.15,0.15}{{6}}}}
\fontsize{10}{0}
\selectfont\put(312.96,244.537){\makebox(0,0)[t]{\textcolor[rgb]{0.15,0.15,0.15}{{8}}}}
\fontsize{10}{0}
\selectfont\put(372.48,244.537){\makebox(0,0)[t]{\textcolor[rgb]{0.15,0.15,0.15}{{10}}}}
\fontsize{10}{0}
\selectfont\put(432,244.537){\makebox(0,0)[t]{\textcolor[rgb]{0.15,0.15,0.15}{{12}}}}
\fontsize{10}{0}
\selectfont\put(491.52,244.537){\makebox(0,0)[t]{\textcolor[rgb]{0.15,0.15,0.15}{{14}}}}
\fontsize{10}{0}
\selectfont\put(69.8755,252.06){\makebox(0,0)[r]{\textcolor[rgb]{0.15,0.15,0.15}{{0}}}}
\fontsize{10}{0}
\selectfont\put(69.8755,280.947){\makebox(0,0)[r]{\textcolor[rgb]{0.15,0.15,0.15}{{5}}}}
\fontsize{10}{0}
\selectfont\put(69.8755,309.835){\makebox(0,0)[r]{\textcolor[rgb]{0.15,0.15,0.15}{{10}}}}
\fontsize{10}{0}
\selectfont\put(69.8755,338.724){\makebox(0,0)[r]{\textcolor[rgb]{0.15,0.15,0.15}{{15}}}}
\fontsize{10}{0}
\selectfont\put(69.8755,367.611){\makebox(0,0)[r]{\textcolor[rgb]{0.15,0.15,0.15}{{20}}}}
\fontsize{10}{0}
\selectfont\put(69.8755,396.5){\makebox(0,0)[r]{\textcolor[rgb]{0.15,0.15,0.15}{{25}}}}
\fontsize{16}{0}
\selectfont\put(298.08,231.537){\makebox(0,0)[t]{\textcolor[rgb]{0.15,0.15,0.15}{{Time}}}}
\fontsize{16}{0}
\selectfont\put(52.8755,324.279){\rotatebox{90}{\makebox(0,0)[b]{\textcolor[rgb]{0.15,0.15,0.15}{{$\| A - \bar{A} \|$}}}}}
\fontsize{9}{0}
\selectfont\put(434.788,382.958){\makebox(0,0)[l]{\textcolor[rgb]{0,0,0}{{Proposed Observer}}}}
\fontsize{9}{0}
\selectfont\put(434.788,368.073){\makebox(0,0)[l]{\textcolor[rgb]{0,0,0}{{Observer by \cite{Chang2021}}}}}
\fontsize{10}{0}
\selectfont\put(74.8799,39.9971){\makebox(0,0)[t]{\textcolor[rgb]{0.15,0.15,0.15}{{0}}}}
\fontsize{10}{0}
\selectfont\put(134.4,39.9971){\makebox(0,0)[t]{\textcolor[rgb]{0.15,0.15,0.15}{{2}}}}
\fontsize{10}{0}
\selectfont\put(193.92,39.9971){\makebox(0,0)[t]{\textcolor[rgb]{0.15,0.15,0.15}{{4}}}}
\fontsize{10}{0}
\selectfont\put(253.44,39.9971){\makebox(0,0)[t]{\textcolor[rgb]{0.15,0.15,0.15}{{6}}}}
\fontsize{10}{0}
\selectfont\put(312.96,39.9971){\makebox(0,0)[t]{\textcolor[rgb]{0.15,0.15,0.15}{{8}}}}
\fontsize{10}{0}
\selectfont\put(372.48,39.9971){\makebox(0,0)[t]{\textcolor[rgb]{0.15,0.15,0.15}{{10}}}}
\fontsize{10}{0}
\selectfont\put(432,39.9971){\makebox(0,0)[t]{\textcolor[rgb]{0.15,0.15,0.15}{{12}}}}
\fontsize{10}{0}
\selectfont\put(491.52,39.9971){\makebox(0,0)[t]{\textcolor[rgb]{0.15,0.15,0.15}{{14}}}}
\fontsize{10}{0}
\selectfont\put(69.8755,47.52){\makebox(0,0)[r]{\textcolor[rgb]{0.15,0.15,0.15}{{0}}}}
\fontsize{10}{0}
\selectfont\put(69.8755,68.1543){\makebox(0,0)[r]{\textcolor[rgb]{0.15,0.15,0.15}{{5}}}}
\fontsize{10}{0}
\selectfont\put(69.8755,88.7886){\makebox(0,0)[r]{\textcolor[rgb]{0.15,0.15,0.15}{{10}}}}
\fontsize{10}{0}
\selectfont\put(69.8755,109.423){\makebox(0,0)[r]{\textcolor[rgb]{0.15,0.15,0.15}{{15}}}}
\fontsize{10}{0}
\selectfont\put(69.8755,130.057){\makebox(0,0)[r]{\textcolor[rgb]{0.15,0.15,0.15}{{20}}}}
\fontsize{10}{0}
\selectfont\put(69.8755,150.691){\makebox(0,0)[r]{\textcolor[rgb]{0.15,0.15,0.15}{{25}}}}
\fontsize{10}{0}
\selectfont\put(69.8755,171.326){\makebox(0,0)[r]{\textcolor[rgb]{0.15,0.15,0.15}{{30}}}}
\fontsize{10}{0}
\selectfont\put(69.8755,191.96){\makebox(0,0)[r]{\textcolor[rgb]{0.15,0.15,0.15}{{35}}}}
\fontsize{16}{0}
\selectfont\put(298.08,26.9971){\makebox(0,0)[t]{\textcolor[rgb]{0.15,0.15,0.15}{{Time}}}}
\fontsize{16}{0}
\selectfont\put(52.8755,119.74){\rotatebox{90}{\makebox(0,0)[b]{\textcolor[rgb]{0.15,0.15,0.15}{{$\| b - \bar{b} \|$}}}}}
\fontsize{9}{0}
\selectfont\put(434.788,175.032){\makebox(0,0)[l]{\textcolor[rgb]{0,0,0}{{Proposed Observer}}}}
\fontsize{9}{0}
\selectfont\put(434.788,160.147){\makebox(0,0)[l]{\textcolor[rgb]{0,0,0}{{Observer by \cite{Chang2021}}}}}
\end{picture}